\documentclass[10pt]{article}

\usepackage{amsmath}
\usepackage{amsfonts}
\usepackage{amsthm}
\usepackage[english]{babel}

\newtheorem{theorem}{Theorem}
\newtheorem{lemma}{Lemma}
\newtheorem{definition}{Definition}
\newtheorem{corollary}{Corollary}

\newtheorem{remark}{Remark}

\newcommand{\nd}{{\rm and }}

\newcommand{\mR}{\mathbb{R}}
\newcommand{\mC}{\mathbb{C}}
\newcommand{\mN}{\mathbb{N}}
\newcommand{\mE}{\mathbb{E}}

\newcommand{\mS}{\mathbb{S}}

\newcommand{\mB}{\mathbb{B}}

\newcommand{\cH}{\mathcal{H}}

\newcommand{\cF}{\mathcal{F}}
\newcommand{\cP}{\mathcal{P}}

\newcommand{\cS}{\mathcal{S}}

\newcommand{\ux}{\underline{x}}
\newcommand{\uxb}{\underline{x} \grave{}}

\newcommand{\uyb}{\underline{y} \grave{}}
\newcommand{\uy}{\underline{y}}

\begin{document}
\title{Orthosymplectically invariant functions in superspace}

\author{K.\ Coulembier\thanks{Ph.D. Fellow of the Research Foundation - Flanders (FWO), E-mail: {\tt Coulembier@cage.ugent.be}}, H.\ De Bie\thanks{Postdoctoral Fellow of the Research Foundation - Flanders (FWO), E-mail: {\tt Hendrik.DeBie@UGent.be}}, F.\ Sommen\thanks{E-mail:{\tt fs@cage.ugent.be}}}

\date{\small{Department of Mathematical Analysis}\\
\small{Faculty of Engineering -- Ghent University\\ Krijgslaan 281, 9000 Gent,
Belgium}}

\maketitle

\begin{abstract}
The notion of spherically symmetric superfunctions as functions invariant under the orthosymplectic group is introduced. This leads to dimensional reduction theorems for differentiation and integration in superspace. These spherically symmetric functions can be used to solve orthosymplectically invariant Schr\"odinger equations in superspace, such as the (an)harmonic oscillator or the Kepler problem. Finally the obtained machinery is used to prove the Funk-Hecke theorem and Bochner's relations in superspace.

\end{abstract}

\newpage

\section{Introduction}

In recent work, we have been developing a new approach to the study of superspace, namely by means of harmonic and Clifford analysis (see e.g. \cite{CDBS1, DBS9, DBE1, DBS5}). We consider a superspace $\mR^{m|2n}$ generated by $m$ commuting or bosonic variables and $2n$ anti-commuting or fermionic variables (\cite{MR732126}). The main feature of this approach is the introduction of an orthosymplectic super Laplace operator $\nabla^2$ and a generalized norm squared $R^2$. They have the property $\nabla^2 (R^2)=2(m-2n)=2M$ with $M$ the so-called super-dimension. As we will see this parameter characterizes several global features of the superspace $\mR^{m|2n}$, see also \cite{DBE1, DBS3, DBS5, MR2492638}. In \cite{CDBS1, DBE1, DBS5} integration over the supersphere (algebraically defined by $R^2=1$) was introduced providing a new and powerful tool in the study of super analysis. 

The aim of this paper is to define, characterize and apply spherically symmetric or orthosymplectically invariant superfunctions. They are defined as functions of the generalized norm or equivalently as functions invariant under the action of the super Lie algebra $\mathfrak{osp}(m|2n)$. We show that these functions have the expected behavior with respect to the differential operators and integration over the supersphere. In particular we obtain dimensional reduction theorems for Berezin integrals (extending results from \cite{MR2492638}) and for differentiation. The fundamental solution for the super Laplace operator with respect to Berezin integration was calculated in \cite{DBS6}. We show that the fundamental solution is an example of a spherically symmetric function. Moreover, this leads to a mean value theorem for super harmonic functions.

Schr\"odinger equations in superspace were considered first as a method to incorporate spin (\cite{MR1019514}). The quantum (an-)harmonic oscillator (\cite{DBS3, MR1019514, MR967935}), the Kepler problem (\cite{MR2395482}), the delta potential (\cite{DBS8}) and the Calogero-Moser-Sutherland-model (\cite{MR2025382}) have already been studied in superspace. We show how general orthosymplectically invariant Schr\"odinger equations (such as the various oscillators and the Kepler problem) can be solved using spherically symmetric functions and spherical harmonics in superspace.

In \cite{DBS5} a super Funk-Hecke formula for polynomials was constructed. Using the supersphere integration from \cite{CDBS1} we can extend this Funk-Hecke theorem to general super zonal functions. This construction leads to Bochner's periodicity relations for the Fourier transform of spherically symmetric functions weighted with spherical harmonics. This, in turn, is equivalent with a super Mehler formula, which is related to the Mehler formula obtained in \cite{CDBS2} via the Hille-Hardy identity, see \cite{HH}.

The paper is organized as follows. After a brief review of the necessary operators on the orthosymplectic Riemannian superspace we define spherically symmetric superfunctions. Then their most important properties with respect to integration and derivation are proven. Next, spherically symmetric functions are used to prove a mean value theorem for harmonic functions and to solve a certain class of Schr\"odinger equations. Finally, the Funk-Hecke theorem for super zonal functions is proven, leading to the Bochner's relation for the super Fourier transform.

\section{Preliminaries}
\label{preliminaries}
Superspaces are spaces where one considers not only commuting (bosonic) but also anti-commuting (fermionic) co-ordinates (see a.o. \cite{MR732126}). The $2n$ anti-commuting variables ${x\grave{}}_i$ generate the complex Grassmann algebra $\Lambda_{2n}$. An arbitrary element $f \in \Lambda_{2n}$ can always be written as $f = \sum_A f_A {x \grave{}}_A$ with ${x \grave{}}_A = {x \grave{}}_1^{\, \alpha_1} \ldots {x \grave{}}_{2n}^{\, \alpha_{2n}}$, $A = (\alpha_{1}, \ldots, \alpha_{2n}) \in \{0,1\}^{2n}$ and $f_A \in \mC$. The dimension of $\Lambda_{2n}$ as a $\mC$-vectorspace is hence $2^{2n}$. We consider a space with $m$ bosonic variables $x_i$. The supervector $\bold{x}$ is defined as
\[
\bold{x}=(X_1,\cdots,X_{m+2n})=(\ux,\uxb)=(x_1,\cdots,x_m,{x\grave{}}_1,\cdots,{x\grave{}}_{2n}).
\]
The corresponding superspace is denoted $\mR^{m|2n}$ and we will always assume $m\not=0$. We consider a Riemannian superspace with orthosymplectic metric. In that case, the inner product of two supervectors $\bold{x}$ and $\bold{y}$ is given by
\begin{eqnarray}
\label{inprod}
\langle \bold{x},\bold{y}\rangle=\langle \ux,\uy\rangle+\langle\uxb,\uyb\rangle&=&\sum_{i=1}^mx_iy_i-\frac{1}{2}\sum_{j=1}^n({x\grave{}}_{2j-1}{y\grave{}}_{2j}-{x\grave{}}_{2j}{y\grave{}}_{2j-1}).
\end{eqnarray}
The commutation relations for two supervectors are determined by $X_iY_j=(-1)^{[i][j]}Y_jX_i$ with $[i]=0$ if $i\le m$ and $[i]=1$ otherwise. This implies that the ${x\grave{}}_{i}$ and ${y\grave{}}_{i}$ together generate the Grassmann algebra $\Lambda_{4n}$. This also means that the inner product \eqref{inprod} is symmetric, i.e. $\langle \bold{x},\bold{y}\rangle=\langle\bold{y},\bold{x}\rangle$. The inner product can be written as $\langle \bold{x},\bold{y}\rangle = \sum_{ij}X_ig^{ij}Y_j$ with metric $g$ defined as
\begin{equation*}
\begin{cases} 
g^{ii}=1 & 1\le i\le m,\\
g^{2i-1+m,2i+m}=-1/2 & 1\le i\le n ,\\  
g^{2i+m,2i-1+m}=1/2 & 1\le i\le n,\\
g^{ij}=0 & otherwise. 
\end{cases}
\end{equation*}

We define $X^j=\sum_iX_ig^{ij}$ so for the commuting variables $x^j=x_j$ and for the anticommuting variables ${x\grave{}}^{2j-1}=\frac{1}{2}{x\grave{}}_{2j}$ and ${x\grave{}}^{2j}=-\frac{1}{2}{x\grave{}}_{2j-1}$. Considering the symmetry of the metric, $g^{ij}=(-1)^{[i]}g^{ji}$, we subsequently obtain 
\begin{equation}
\label{inproduit}
\langle \bold{x},\bold{y}\rangle=\sum_jX^jY_j=\sum_{i,j}X_ig^{ji}(-1)^{[i]}Y_j=\sum_i(-1)^{[i]}X_iY^i.
\end{equation}
The generalized norm squared is given by
\begin{equation*}
R^2=\langle \bold{x},\bold{x}\rangle=\sum_{j=1}^{m+2n}X^jX_j =\sum_{i=1}^mx_i^2-\sum_{j=1}^n{x\grave{}}_{2j-1}{x\grave{}}_{2j}=r^2-\uxb^2.
\end{equation*}
This  $R^2$ is used in the study of certain quantum hamiltonians in superspace (see \cite{DBS8, DBS3, MR967935, MR2395482}). In previous papers (e.g. \cite{CDBS1, CDBS2, DBE1, DBS3, DBS5}) we used the notation $\bold{x}^2=-R^2$. The fermionic norm squared satisfies
\begin{eqnarray*}
\uxb^{2n}=n!{x\grave{}}_{1}{x\grave{}}_{2}\cdots {x\grave{}}_{2n},
\end{eqnarray*}
which is the element of maximal degree $2n$ in $\Lambda_{2n}$.

The fermionic partial derivatives $\partial_{{x\grave{}}_{j}}$ commute with the bosonic variables and satisfy the Leibniz rule $\partial_{{x\grave{}}_{j}}{x\grave{}}_{k}=\delta_{jk}-{x\grave{}}_{k}\partial_{{x\grave{}}_{j}}$. The super gradient is defined as
\begin{eqnarray*}
\nabla&=&(\partial_{X^1},\cdots,\partial_{X^{m+2n}}).
\end{eqnarray*}
Using $\partial_{{x\grave{}}^{2j-1}}=2\partial_{{x\grave{}}_{2j}}$ we find 
\begin{eqnarray*}
\nabla&=&(\partial_{x_1},\cdots,\partial_{x_m},2\partial_{{x\grave{}}_{2}},-2\partial_{{x\grave{}}_{1}},\cdots,2\partial_{{x\grave{}}_{2n}},-2\partial_{{x\grave{}}_{2n-1}}),
\end{eqnarray*}
and $\nabla^j=(-1)^{[j]}\partial_{X_j}$. The super Laplace operator is given by
\begin{equation*}
\nabla^2=\langle \nabla,\nabla \rangle=\sum_{k=1}^{m+2n}\nabla^k\nabla_k=\sum_{k=1}^{m+2n}(-1)^{[k]}\partial_{X_k}\partial_{X^k} =\sum_{i=1}^m\partial_{x_i}^2-4\sum_{j=1}^n\partial_{{x\grave{}}_{2j-1}}\partial_{{x\grave{}}_{2j}}=\nabla^2_b+\nabla^2_f.
\end{equation*}

The super Euler operator is defined as
\begin{equation}
\label{Euler}
\mE=\mE_b+\mE_f=\langle \bold{x},\nabla\rangle=\sum_{k=1}^{m+2n}X^k\nabla_{k}=\sum_{k=1}^{m+2n}X_k\partial_{X_k}=\sum_{i=1}^mx_i\partial_{x_i}+\sum_{j=1}^{2n}{x\grave{}}_j\partial_{{x\grave{}}_j}.
\end{equation}

The operators $i\nabla^2/2$, $iR^2/2$ and $\mE+M/2$, generate the $\mathfrak{sl}_2$ Lie-algebra, see \cite{DBE1, DBS5}. This is a consequence of the commutators
\begin{eqnarray*}
\left[\nabla^2/2,R^2/2\right]&=&\mE+M/2\\
\left[\nabla^2/2,\mE+M/2\right]&=&2\nabla^2/2\\
\left[R^2/2,\mE+M/2\right]&=&-2R^2/2.
\end{eqnarray*} 
The following calculation also extends the bosonic case:
\begin{eqnarray}
\label{inprodxnabla}
\langle \nabla,\bold{x}\rangle&=&M+\mE.
\end{eqnarray}

The Laplace-Beltrami operator is defined as
\begin{eqnarray}
\label{LB}
\Delta_{LB}&=&R^2\nabla^2-\mE(M-2+\mE).
\end{eqnarray}
The orthosymplectic Lie superalgebra $\mathfrak{osp}(m|2n)$ is generated by the following differential operators in superspace (see \cite{MR2395482})
\begin{eqnarray}
\label{ospgen}
L_{ij}&=&X_i\partial_{X^j}-(-1)^{[i][j]}X_j\partial_{X^i}
\end{eqnarray}
for $1\le i \le j \le m+2n$. Note that in the orthosymplectic case $L_{jj}$ is not necessarily zero, since e.g. $L_{m+2i,m+2i}=-4{x\grave{}}_{2i}\partial_{{x\grave{}}_{2i-1}}$. A calculation shows that the Laplace-Beltrami operator can be expressed as
\begin{eqnarray*}
\Delta_{LB}&=&-\frac{1}{2}\sum_{i,j,k,l=1}^{m+2n}L_{ij}g^{ik}g^{jl}L_{kl},
\end{eqnarray*} 
see \cite{MR0546778}. Since the Laplace operator is orthosymplectically invariant and the Euler operator commutes with the $L_{ij}$, the Laplace-Beltrami operator also commutes with all the $L_{ij}$. So we find that the Laplace-Beltrami operator is a Casimir operator of $\mathfrak{osp}(m|2n)$ of degree $2$, see \cite{MR1773773, MR0546778}.

The space of superpolynomials is given by $\cP=\mR[x_1,\cdots,x_m]\otimes \Lambda_{2n}$. More general superfunctions can for instance be defined as functions with values in the Grassmann algebra, $f:\Omega\subset\mR^m\to\Lambda_{2n}$. They can always be expanded as 
\begin{eqnarray}
f&=&\sum_A {x\grave{}}_A f_A,
\label{superftie}
\end{eqnarray}
with ${x\grave{}}_A$ the basis of monomials for the Grassmann algebra and $f_A:\Omega\subset\mR^m\to\mC$. The purely bosonic part $f_0$, with ${x\grave{}}_0=1$ is usually called the body of the superfunction. In general, for a function space $\cF$ corresponding to the $m$ bosonic variables (e.g. $\cS(\mR^{m})$, $L_p(\mR^m)$, $C^k(\Omega)$) we use the notation $\cF_{m|2n}=\cF\otimes\Lambda_{2n}$.

The null-solutions of the super Laplace operator are called harmonic superfunctions. In particular we are interested in harmonic superpolynomials.
\begin{definition}
An element $F \in \cP$ is a spherical harmonic of degree $k$ if it satisfies
\begin{eqnarray*}
\nabla^2 F &=&0\\
\mE F &=& kF, \quad \mbox{i.e. $F \in \cP_k$}.
\end{eqnarray*}
Moreover the space of all spherical harmonics of degree $k$ is denoted by $\cH_k$.
\end{definition}

In the purely bosonic case we denote this space by $\cH_k^b$. Since we will use this space for different dimensions, we include the number of bosonic variables $d$ in the notation, $\cH_{k,d}^b$. Formula \eqref{LB} implies that spherical harmonics are eigenfunctions of the Laplace-Beltrami operator,
\begin{eqnarray}
\label{LBH}
\Delta_{LB}H_k&=&-k(M-2+k)H_k,
\end{eqnarray}
for $H_k\in\cH_k$. We have the following decomposition (see \cite{DBS5}).

\begin{lemma}[Fischer decomposition]
If $M \not \in -2 \mN$, $\cP$ decomposes as
\begin{equation*}
\cP = \bigoplus_{k=0}^{\infty} \cP_k= \bigoplus_{j=0}^{\infty} \bigoplus_{k=0}^{\infty} R^{2j}\cH_k.
\end{equation*}
\label{scalFischer}
\end{lemma}

This decomposition captures the essence of the Howe dual pair $(\mathfrak{osp}(m|2n),\mathfrak{sl}_2)$ with $\mathfrak{sl}_2$ generated by $i\nabla^2/2$ and $iR^2/2$ and $\mathfrak{osp}(m|2n)$ generated by the $L_{ij}$ in equation \eqref{ospgen}. Each $\bigoplus_{j=0}^{\infty}  R^{2j}\cH_k$ is an irreducible $\mathfrak{sl}_2$-representation, the weight vectors are $R^{2j}\cH_k$. The blocks $R^{2j}\cH_k$ are exactly the irreducible pieces of the representation of $\mathfrak{osp}(m|2n)$ on $\cP$ when $M>0$, see \cite{MR2395482}.

In the purely fermionic case ($m=0$) this decomposition is
\begin{eqnarray}
\label{fermFischer}
\Lambda_{2n}&=&\bigoplus_{j=0}^n\bigoplus_{k=0}^{n-j}\uxb^{2j}\cH_k^f,
\end{eqnarray}
which corresponds to the decomposition of $\Lambda_{2n}$ into irreducible pieces under the action of $\mathfrak{sp}(2n)$.

The integration used on $\Lambda_{2n}$ is the so-called Berezin integral (see \cite{MR732126,DBS5}), defined by
\begin{eqnarray}
\label{Berezin}
\int_{B}& =& \pi^{-n} \partial_{{x \grave{}}_{2n}} \ldots \partial_{{x \grave{}}_{1}} = \frac{ \pi^{-n}}{4^n n!} \nabla_f^{2n}.
\end{eqnarray}
On a general superspace, integration is then defined by
\begin{equation}
\label{superint}
\int_{\mR^{m | 2n}} = \int_{\mR^m} dV(\ux)\int_B=\int_B \int_{\mR^m} dV(\ux),
\end{equation}
with $dV(\ux)$ the usual Lebesgue measure in $\mR^{m}$. Note that we omit the measure in the notation for integration over superspace.

The supersphere is algebraically defined by the relation $R^2=1$, when $m\not=0$. The integration over the supersphere was introduced in \cite{DBS5} for polynomials and generalized to a broader class of functions in \cite{CDBS1}. The integration for polynomials is uniquely defined by the following properties, see \cite{CDBS1, DBE1}.

\begin{theorem}
When $m\not=0$, the only (up to a multiplicative constant) linear functional $T: \cP \rightarrow \mR$ satisfying the following properties for all $f(\bold{x}) \in \cP$:
\begin{itemize}
\item $T(R^2 f(\bold{x})) = T(f(\bold{x}))$
\item $T(f(g \cdot \bold{x})) = T(f(\bold{x}))$, \quad $\forall g \in SO(m)\times Sp(2n)$
\item $k \neq l \quad \Longrightarrow \quad T(\cH_k \cH_l) = 0 = T(\cH_l \cH_k)$
\end{itemize}
is given by the Pizzetti integral
\begin{equation}
\label{Pizzetti}
\int_{SS} f(\bold{x})  =  \sum_{k=0}^{\infty}  \frac{2 \pi^{M/2}}{2^{2k} k!\Gamma(k+M/2)} (\nabla^{2k} f )(0),\qquad f(\bold{x})\in\cP.
\end{equation}
\end{theorem}

For a general superfunction which is $n$ times continuously differentiable in an open neighborhood of the unit sphere $\mS^{m-1}$, the integration over the supersphere is defined as (see \cite{CDBS1}, theorem 8)

\begin{eqnarray}
\label{integSS}
\int_{SS}f&=&\sum_{j=0}^n\int_{\mS^{m-1}}d\underline{\xi}\int_B\frac{\uxb^{2j}}{j!}\left[(\frac{\partial}{\partial r^2})^jr^{m-2}f\right]_{r=1},
\end{eqnarray}
using spherical co-ordinates $\ux = r \underline{\xi}$ and with $d \underline{\xi}$ the surface measure on $\mS^{m-1}$. For a polynomial, formulas \eqref{Pizzetti} and \eqref{integSS} coincide. It should be stressed that not all terms in the expansion \eqref{superftie} of a superfunction $f$ need to be $n$ times differentiable. The body has to be $n$ times differentiable but for instance the part $f_{1,\cdots,1}$ only has to be an element of $L_1(\mS^{m-1})$. The area of the supersphere is given by
\begin{eqnarray*}
\sigma_M=\int_{SS}1&=&\frac{2\pi^{M/2}}{\Gamma(M/2)}.
\end{eqnarray*}

The supersphere integration \eqref{Pizzetti} is an immediate generalization of the bosonic Pizzetti formula, by substitution of the dimension. The Berezin integral can be connected with the supersphere integration again by dimensional continuation, see \cite{CDBS1}. This gives a new logical interpretation of the Berezin integral.

\begin{theorem}
\label{connectieBer}
For $M>0$ and $f$ a function in $L_1(\mR^m)_{m|2n}\cap C^n(\mR^m)_{m|2n}$, the following relation holds
\begin{eqnarray*}
\int_{\mR^{m|2n}}f&=&\int_{0}^\infty dv\, v^{M-1}\int_{SS,\bold{x}}f(v\bold{x}).
\end{eqnarray*}
\end{theorem}

This theorem already reveals that we can expect dimensional reduction for evaluations of Berezin integrals, as e.g. obtained in \cite{MR2492638}.

Integration over the superball is defined as
\begin{eqnarray*}
\int_{SB}f&=&\int_{\mB^m}dV(\ux)\int_Bf+\sum_{j=0}^{n-1}\int_{\mS^{m-1}}d\underline{\xi}\int_B\frac{\uxb^{2j+2}}{(j+1)!}\left[(\frac{\partial}{\partial r}\frac{1}{2r})^jr^{m-1}f\right]_{r=1},
\end{eqnarray*}
where $\mB^m$ is the unit ball in $\mR^m$. Integration over the supersphere is connected with integration over the superball by Green's theorem (see \cite{CDBS1, DBS5}).

\begin{lemma}
\label{Cauchysphere}
For a superfunction $f$ which is $n$ times continuously differentiable in an open neighborhood of the unit sphere $\mS^{m-1}$ and for wich $\nabla f$ exists almost everywhere and is integrable over $\mB^m$, the following holds
\begin{eqnarray*}
\int_{SB}\nabla f&=&\int_{SS}\bold{x}f.
\end{eqnarray*}
\end{lemma}
In other words, the supervector $\bold{x}$ acts as an outer normal on the supersphere. This integral formula differs from the Stokes' formulas on superspace in \cite{MR0647158} and \cite{MR1402921} since it does not connect a domain on $\mR^{m|2n}$ with its boundaries of codimension $(1,0)$ and $(0,1)$, see \cite{CDBS1}.

The Clifford-Hermite functions are eigenvectors of the isotropic harmonic oscillator in $\mR^m$, for an overview see \cite{CDBS2}. They provide an alternative for the basis of products of Hermite functions, based on the $SO(m)$-invariance of the hamiltonian of the harmonic oscillator. In superspace they were defined in \cite{DBS3} and studied in full detail in \cite{CDBS2}. They can be expressed as
\begin{eqnarray}
\label{CH}
\psi_{j,k,l}&=&2^{2j}j!L_j^{\frac{M}{2}+k-1}(R^2)H_k^{(l)}\exp(-R^2/2),
\end{eqnarray}
with $H_k^{(l)}$ a basis of $\cH_k$ and $L_p^q$ the generalized Laguerre polynomials. Here, the super Gaussian is defined by
\[
\exp(-R^2/2)=\exp(-r^2/2)\exp(\uxb^2/2)=\exp(-r^2/2)\sum_{j=0}^n\frac{\uxb^{2j}}{2^jj!}.
\]
For the purely bosonic one-dimensional case we obtain $\psi_{j,0,1}=2^{2j}j!L_j^{-\frac{1}{2}}(x^2)\exp(-x^2/2)$ and $\psi_{j,1,1}=2^{2j}j!L_j^{\frac{1}{2}}(x^2)x\exp(-x^2/2)$ which are the classical Hermite functions on the real line. When $M\not\in-2\mN$, the super Clifford-Hermite functions in formula (\ref{CH}) form a basis of the space $\cP\exp(-R^2/2)$, see lemma \ref{scalFischer}. They are the eigenvectors of the super harmonic oscillator,
\begin{eqnarray*}
\frac{1}{2}\left(R^2-\nabla^2\right)\psi_{j,k,l}&=&(2j+k+\frac{M}{2})\psi_{j,k,l}.
\end{eqnarray*}

In \cite{DBS9} the super Fourier transform on $\cS(\mR^m)_{m|2n}$ was introduced as
\begin{eqnarray}
\label{Four}
\cF^{\pm}_{m|2n}(f(\bold{x}))(\bold{y})&=&(2\pi)^{-M/2}\int_{\mR^{m|2n},x}\exp(\pm i\langle \bold{x},\bold{y}\rangle)f(\bold{x}),
\end{eqnarray}
yielding an $SO(m)\times Sp(2n)$-invariant generalization of the purely bosonic Fourier transform. In \cite{CDBS2} the following Mehler formula for the super Fourier kernel was proven,
\begin{equation}
\label{superMehler2}
\sum_{j,k}\frac{2j!(-1)^j (\pm i)^{k}}{\Gamma(j+\frac{M}{2}+k)}    L_j^{\frac{M}{2}+k-1}(R^2_x)L_j^{\frac{M}{2}+k-1}(R_y^2) F_k (\bold{x},\bold{y}) \exp(-\frac{R^2_x + R^2_y}{2})
= \left(2\pi \right)^{-\frac{M}{2}} \exp(\pm i \langle \bold{x} ,\bold{ y} \rangle).
\end{equation}
In this formula, $F_k(\bold{x},\bold{y})$ is the reproducing kernel of $\cH_k$, satisfying
\begin{eqnarray}
\label{reprkern}
\int_{SS,\bold{x}} H_l(\bold{x}) F_{k}(\bold{x},\bold{y}) &=& \delta_{kl} H_l(\bold{y}), \quad \mbox{for all } H_l \in \cH_l.
\end{eqnarray}
$F_k(\bold{x},\bold{y})$ is explicitly given by
\begin{eqnarray*}
F_k(\bold{x},\bold{y}) &=& \frac{2k+M-2}{M-2} \frac{\Gamma(M/2)}{2 \pi^{M/2}} (R_xR_y)^{k} C^{(M-2)/2}_k \left(\frac{\langle  \bold{x},\bold{y} \rangle}{R_xR_y} \right).
\end{eqnarray*}
with $C^{(M-2)/2}_k(t)$ a Gegenbauer polynomial. Since the Gegenbauer polynomial is even (odd) when $k$ is even (odd), only positive and even powers of $R_x$ and $R_y$ appear and the expression is well-defined. 

\section{Orthosymplectically invariant superfunctions}

We want to consider functions which are spherically symmetric in superspace. They are called spherically symmetric or orthosymplectically invariant superfunctions. In stead of starting from the definition as $\mathfrak{osp}(m|2n)$-invariant functions, we make the analogy with bosonic analysis and define them as functions of `$R$', using a suitable Taylor expansion in the fermionic variables. We will prove later that these functions are exactly the $\mathfrak{osp}(m|2n)$-invariant functions.

\begin{definition}
\label{defftien}
For a function $h:\mR^+\to\mR$, $h\in C^n(\mR^+)$, the superfunction with notation $h(R^2)$ is defined as
\begin{eqnarray*}
h(R^2)&=&\sum_{j=0}^n(-1)^j\frac{\uxb^{2j}}{j!}h^{(j)}(r^2)
\end{eqnarray*}
and is an element of $C^{0}(\mR^m)_{m|2n}$.
\end{definition}

It is clear that $\mR^+$ can be replaced by any open $\Sigma\subset\mR^+$ in this definition. As these functions only contain even elements of the Grassmann algebra, spherically symmetric functions commute with all elements of $\Lambda_{2n}$. The super Gaussian introduced in equation \eqref{CH} is an example of definition \ref{defftien}.

The space of spherically symmetric superfunctions also forms a commutative $\mC$-algebra with multiplication the usual multiplication of superfunctions. The fact that the space is closed with respect to this multiplication follows immediately from 
\begin{equation}
g(R^2)h(R^2)=(gh)(R^2),
\label{algmor}
\end{equation}
for $g,h\in C^n(\mR^+)$ and the notation of definition \ref{defftien}. This also implies that definition \ref{defftien} generates a $\mC$-algebra morphism between the algebra of functions $C^n(\mR^+)$ and the algebra of spherically symmetric functions.

As an example we define the superfunction $\sqrt{R^2}$ and call it $R$,
\begin{eqnarray*}
R&=&\sum_{j=0}^n(-1)^j\frac{\uxb^{2j}}{j!}\frac{\Gamma(\frac{3}{2})}{\Gamma(\frac{3}{2}-j)}r^{1-2j}.
\end{eqnarray*}

We can extend definition \ref{defftien} to more general arguments than $R^2$. 
\begin{definition}
\label{defftien2}
Let $h$ be a function $h:\mR\to\mR$, $h\in C^{2n}(\mR)$. For every superfunction $f(\bold{x})=f_0(\ux)+f_1(\bold{x})$ (as in formula (\ref{superftie})), with body $f_0$ and nilpotent part $f_1$, the superfunction with notation $h(f(\bold{x}))$ is defined as
\begin{eqnarray*}
h(f(\bold{x}))&=&\sum_{j=0}^{2n}\frac{(f_1(\bold{x}))^{j}}{j!}h^{(j)}(f_0(\ux)).
\end{eqnarray*}
\end{definition}

This definition allows to rewrite the spherically symmetric superfunction corresponding to the composition of two functions $h,g\in C^n(\mR^+)$, 
\begin{eqnarray}
\label{circ}
(h\circ g)(R^2)&=&h(g(R^2)).
\end{eqnarray}
We hence obtain an equivalence between functions of the form $h(R)$ (defined as in definition \ref{defftien2} with $f(\bold{x})=R$) and spherically symmetric functions. This is given by $h(R)=g(R^2)$ with $g,h\in C^n(\mR^+)$ satisfying $g=h\circ\sqrt{\cdot}$. These functions $h(R)$ could therefore also have been used as the definition for spherically symmetric functions. This means we can define the function $R^{\alpha}$, for $\alpha\in\mR$ using definition \ref{defftien2} as a function of the superfunction $R$, or using equation (\ref{circ}) and definition \ref{defftien} as $R^{\alpha}=(R^2)^{\frac{\alpha}{2}}$. Both definitions are equal and give
\begin{eqnarray*}
R^\alpha&=&\sum_{j=0}^n(-1)^j\frac{\uxb^{2j}}{j!}(\frac{\alpha}{2}+1-j)_jr^{\alpha-2j},
\end{eqnarray*}
with $(a)_j=(a)(a+1)\cdots (a+j-1)$ the Pochhammer symbol. In case $\alpha\in2\mN$, this reduces to ordinary polynomials $R^{2k}$ using the binomial theorem. 

Now we will show that the spherically symmetric functions have the expected behavior. A straightforward calculation yields the following lemma
\begin{lemma}
\label{diracaxiaal}
With the notations from definition \ref{defftien} and $h\in C^{n+1}(\mR^+)$, the following holds,
\begin{eqnarray*}
\nabla h(R^2)&=&2\bold{x}h^{(1)}(R^2).
\end{eqnarray*}
\end{lemma}
\begin{proof}
The bosonic derivatives follow from $\partial_{x_i}h^{(j)}(r^2)=2x_ih^{(j+1)}(r^2)$. For the fermionic part we calculate
\begin{eqnarray*}
\partial_{{x\grave{}}^{2k-1}}h(R^2)&=&2\partial_{{x\grave{}}_{2k}}\sum_{j=0}^n(-1)^j\frac{\uxb^{2j}}{j!}h^{(j)}(r^2)\\
&=&-2{x\grave{}}_{2k-1}\sum_{j=1}^n(-1)^jj\frac{\uxb^{2j-2}}{j!}h^{(j)}(r^2)\\
&=&2{x\grave{}}_{2k-1}\sum_{j=1}^{n+1}(-1)^{(j-1)}\frac{\uxb^{2j-2}}{(j-1)!}h^{(j)}(r^2)\\
&=&2{x\grave{}}_{2k-1}h^{(1)}(R^2).
\end{eqnarray*}
The term $\partial_{{x\grave{}}^{2k}}h(R^2)$ is calculated similarly.
\end{proof}

Combining lemma \ref{diracaxiaal} with equation (\ref{Euler}) we obtain
\begin{eqnarray}
\label{Eulerspherical}
\mE h(R^2)&=&2R^2 h^{(1)}(R^2).
\end{eqnarray}
Therefore, the super Euler operator formally acts as $R\partial_{R}$ on spherically symmetric functions. 

Lemma \ref{diracaxiaal} implies that spherically symmetric functions are null-solutions of all the differential operators $L_{ij}$ in equation \eqref{ospgen}. We can prove that all the functions which are null-solutions of the generators of $\mathfrak{osp}(m|2n)$ are of this form.

\begin{theorem}
\label{uniciteit}
A differentiable superfunction $f(\bold{x})$ which is orthosympectically invariant, i.e. for which
\begin{eqnarray*}
L_{ij}f(\bold{x})&=&0
\end{eqnarray*}
for all $i\le j$ is a spherically symmetric superfunction of the form of definition \ref{defftien}.
\end{theorem}
\begin{proof}
It was already mentioned that spherically symmetric functions are orthosymplectically invariant. Now we prove the other direction of the theorem. We know that a purely bosonic solution of the bosonic $L_{ij}$ is of the form $g(r^2)$ with $g$ differentiable. Elements of $\Lambda_{2n}$ which are null-solutions of all $X\in\mathfrak{sp}(2n)$, have to be contained in the one dimensional blocks in the $\mathfrak{sp}(2n)$-decomposition \eqref{fermFischer} of $\Lambda_{2n}$. These correspond to $\uxb^{2j}$.

From these considerations we conclude that $f$ has to be of the form
\begin{eqnarray*}
f(\bold{x})&=&\sum_{k=0}^n\frac{(-1)^k}{k!} f_k(r^2)\uxb^{2k}
\end{eqnarray*}
for some differentiable functions $f_k$. Now applying the mixed (bosonic-fermionic) $L_{ij}$,
\begin{eqnarray*}
(2x_j\partial_{{x\grave{}}_{2i}}-{x\grave{}}_{2i-1}\partial_{x_j})f(\bold{x})&=&\sum_{k=1}^n2x_j\frac{(-1)^k}{k!}kf_k(r^2)(-{x\grave{}}_{2i-1})\uxb^{2k-2}-\sum_{k=0}^{n-1}2x_j\frac{(-1)^k}{k!}f_k'(r^2){x\grave{}}_{2i-1}\uxb^{2k}\\
&=&\sum_{k=0}^{n-1}2x_j\frac{(-1)^k}{k!}f_{k+1}(r^2){x\grave{}}_{2i-1}\uxb^{2k}-\sum_{k=0}^{n-1}2x_j\frac{(-1)^k}{k!}f_k'(r^2){x\grave{}}_{2i-1}\uxb^{2k}\\
&=&0,
\end{eqnarray*}
we finally find $f'_k=f_{k+1}$.
\end{proof}

Now we prove `dimensional reduction' theorems for differentiating spherically symmetric functions. This means that most of the behavior of the functions is captured in the super-dimension $M$. So in case $M>0$, we can use formulas well-known for bosonic analysis in $M$ dimensions. We start with the expression of the Laplace operator on spherically symmetric functions. In $m$ bosonic dimensions this is given by the formula,

\begin{equation}
\label{Laplbosrad}
\nabla^2_b h(r^2)=(\frac{\partial^2}{\partial r^2}+\frac{m-1}{r}\frac{\partial}{\partial r})h(r^2)=4h^{(2)}(r^2)r^2+2mh^{(1)}(r^2).
\end{equation}

\begin{lemma}
\label{deltaftie}
With the notations from definition \ref{defftien} and $h\in C^{n+2}(\mR^+)$ 
\begin{eqnarray*}
\nabla^2 h(R^2)&=&4h^{(2)}(R^2)R^2+2Mh^{(1)}(R^2).
\end{eqnarray*}
This can be generalized to
\begin{eqnarray}
\label{commLaplRkwad}
[\nabla^2,h(R^2)]&=&4R^2h^{(2)}(R^2)+h^{(1)}(R^2)(4\mE+2M).
\end{eqnarray}
\end{lemma}

\begin{proof}

Using lemma \ref{diracaxiaal} and equations (\ref{inprodxnabla}) and (\ref{Eulerspherical}) we calculate
\begin{eqnarray*}
\nabla^2 h(R^2)&=&2\langle \nabla,x\rangle h^{(1)}(R^2)\\
&=&(2\mE+2M)h^{(1)}(R^2)\\
&=&4R^2h^{(2)}(R^2)+2Mh^{(1)}(R^2).
\end{eqnarray*}
The generalization follows from formula \eqref{inproduit}
\begin{eqnarray*}
[\nabla^2,h(R^2)]&=&\nabla^2(h(R^2))+\sum_{k}\left(\left[\nabla^kh(R^2)\right]\nabla_k+(-1)^{[k]}\left[\nabla_kh(R^2)\right]\nabla^k\right)\\
&=&\nabla^2(h(R^2))+2\langle \nabla(h(R^2)),\nabla\rangle\\
&=&4R^2h^{(2)}(R^2)+2Mh^{(1)}(R^2)+4h^{(1)}(R^2)\mE.
\end{eqnarray*}
\end{proof}

Using equation \eqref{commLaplRkwad} we thus obtain
\begin{eqnarray}
\label{laplradharm}
\nabla^2 h(R^2)H_k&=&4R^2h^{(2)}(R^2)H_k+(4k+2M)h^{(1)}(R^2)H_k,
\end{eqnarray}
for $H_k\in\cH_k$. 

The Laplace-Beltrami operator and the spherically symmetric functions commute.
\begin{lemma}
\label{axiaalLB}
If $h\in C^{n+2}(\mR^+)$, then
\begin{eqnarray*}
[\Delta_{LB},h(R^2)]&=&0.
\end{eqnarray*}
\end{lemma}

\begin{proof}
This follows from lemma \ref{deltaftie} and equations (\ref{Eulerspherical}) and (\ref{LB}). The result can equally be obtained from $[L_{ij},h(R^2)]=0$ because $(L_{ij}h(R^2))=0$ and $L_{ij}$ is a first order derivative.
\end{proof}

Now we prove integral theorems for spherically symmetric functions. The following lemma implies that integration on the supersphere acts on spherically symmetric functions as expected, i.e modulo $R^2-1$.
\begin{lemma}
\label{xkwad}
For $g(\bold{x})$ $n$ times continuously differentiable in an open neighborhood of the unit sphere $\mS^{m-1}$ and $h(R^2)$ as in definition \ref{defftien}, the following holds
\begin{eqnarray*}
\int_{SS}h(R^2)g(\bold{x})&=&h(1)\int_{SS}g(\bold{x}).
\end{eqnarray*}
\end{lemma}

\begin{proof}

We start by calculating
\begin{eqnarray*}
\int_{SS}h(R^2)g&=&\sum_{j=0}^n\int_{\mS^{m-1}}d\underline{\xi}\int_{B}\sum_{k=0}^{n-j}\frac{\uxb^{2j+2k}}{j!k!}(-1)^k\left[(\frac{\partial}{\partial r^2})^jr^{m-2}h^{(k)}(r^2)g\right]_{r=1}.
\end{eqnarray*}
So we find $\int_{SS}h(R^2)g$ is well-defined since the derivatives of $h$ appear only up to $h^{(n)}$. We then find
\begin{eqnarray*}
\int_{SS}h(R^2)g&=&\int_{\mS^{m-1}}d\underline{\xi}\int_{B}\sum_{t=0}^n\sum_{s=0}^{t}\frac{\uxb^{2t}}{t!}\binom{t}{s}(-1)^{t-s}\left[(\frac{\partial}{\partial r^2})^sr^{m-2}g[(\frac{\partial}{\partial r^2})^{t-s}h(r^2)]\right]_{r=1}\\
&=&\sum_{t=0}^n\int_{\mS^{m-1}}d\underline{\xi}\int_{B}\frac{\uxb^{2t}}{t!}\sum_{s=0}^{t}\binom{t}{s}(-1)^{t-s}\sum_{p=0}^s\binom{s}{p}\left[(\frac{\partial}{\partial r^2})^pr^{m-2}g\right]_{r=1}\left[(\frac{\partial}{\partial r^2})^{t-s+s-p}h(r^2)\right]_{r=1}\\
&=&\sum_{t=0}^n\int_{\mS^{m-1}}d\underline{\xi}\int_{B}\frac{\uxb^{2t}}{t!}\sum_{p=0}^{t}\left[(\frac{\partial}{\partial r^2})^pr^{m-2}g\right]_{r=1}\left[(\frac{\partial}{\partial r^2})^{t-p}h(r^2)\right]_{r=1}\sum_{s=p}^{t}\binom{t}{s}\binom{s}{p}(-1)^{t-s}.
\end{eqnarray*}
Since $\sum_{s=p}^{t}\binom{t}{s}\binom{s}{p}(-1)^{t-s}=(-1)^{t}\frac{t!}{p!}\sum_{q=0}^{t-p}\frac{1}{(t-p-q)!q!}(-1)^{q+p}=(-1)^{t+p}\binom{t}{p}\delta_{t,p}=\delta_{t,p}$, we obtain
\begin{eqnarray*}
\int_{SS}h(R^2)g&=&\int_{\mS^{m-1}}d\underline{\xi}\int_{B}\sum_{t=0}^n\frac{\uxb^{2t}}{t!}\left[(\frac{\partial}{\partial r^2})^tr^{m-2}g\right]_{r=1}\left[h(r^2)\right]_{r=1}\\
&=&h(1)\int_{SS}g,
\end{eqnarray*}
which completes the proof.
\end{proof}

In this section we reserve the notation $\uy$ for a bosonic vector in $\mR^M$. Combining lemma \ref{xkwad} with theorem \ref{connectieBer} we find that for $h(R^2)\in L_1(\mR^{m})_{m|2n}$ and $M>0$
\begin{eqnarray*}
\int_{\mR^{m|2n}}h(R^2)&=&\sigma_M\int_{0}^\infty v^{M-1}h(v^2)dv \\
&=&\int_{\mR^{M}}h(r_{\uy}^2)dV(\uy).
\end{eqnarray*}

Hence, in case $M>0$, the Berezin integral of a spherically symmetric function depends only on the super-dimension $M$ and not on $m$ and $2n$ separately. In particular, it is equal to the Lebesgue integral in $M$ bosonic dimensions. This result was also found in \cite{MR2492638} using a different approach. Now we use our approach to calculate the integral in case $M\le 0$ in order to reobtain theorem III.1 in \cite{MR2492638}. Here, we explicitly prove this theorem and derive a sufficient condition on the functions for which it holds. This condition turns out to be $h(\bold{x})\in L_1(\mR^m)_{m|2n}$.

\begin{theorem}
\label{dimredint}
Let $h(\bold{x})$ be a differentiable superfunction in $L_1(\mR^m)_{m|2n}$, satisfying $L_{ij}h=0$ for all $1\le i\le j\le m+2n$ with body $f$, then $h(\bold{x})$ is of the form $f(R^2)$ and
\begin{equation*}
\int_{\mR^{m|2n}}h(\bold{x})=
\begin{cases}
\vspace{2mm}
\int_{\mR^{M}}f(r_{\uy}^2)dV(\uy) & M>0\\
\vspace{2mm}
(-\pi)^{\frac{M}{2}}f^{(-M/2)}(0) & M\in-2\mN\\
2(-\pi)^{\frac{M-1}{2}}\int_{0}^\infty  f^{(\frac{1-M}{2})}(r^2) dr& M\in-2\mN-1.
\end{cases}
\end{equation*}
\end{theorem}

\begin{proof}

The case $M>0$ is already obtained. Now we consider the case $M\le 0$. Since $f(R^2)$ is an element of $L_1(\mR^m)_{m|2n}$ we find that $\lim_{r\to\infty}r^{m-1}f^{(j)}(r^2)=0$ holds for $j=0,\cdots,n$. This implies that
\begin{eqnarray}
\label{randinfty}
\lim_{r\to\infty}r^{M+2j}f^{(j)}(r^2)=0 &\mbox{ for}& j=\lfloor \frac{1-M}{2}\rfloor,\cdots, n-1.
\end{eqnarray}
For a spherically symmetric function we can calculate easily that
\begin{eqnarray*}
\int_Bf(R^2)&=&\frac{(-1)^n}{\pi^n}f^{(n)}(r^2).
\end{eqnarray*}
First we restrict to the case $m>2$ and insert this into the left-hand side of the theorem,
\begin{eqnarray*}
\int_{\mR^{m|2n}}f(R^2)&=&\frac{\sigma_m}{\pi^n}(-1)^n\int_{0}^\infty  r^{m-1}(\frac{1}{2r}\frac{\partial}{\partial r})^n f(r^2)dr\\
&=&-\frac{\sigma_m \frac{m-2}{2}}{\pi^n}(-1)^n\int_{0}^\infty  r^{m-3}(\frac{1}{2r}\frac{\partial}{\partial r})^{n-1} f(r^2)dr+\frac{\sigma_m}{2\pi^n}(-1)^n\left[r^{m-2}f^{(n-1)}(r^2)\right]^\infty_0\\
&=&\int_{\mR^{m-2|2n-2}}f(R_v^2)-\frac{\sigma_m}{2\pi^n}(-1)^n\lim_{r\to_>0}r^{m-2}f^{(n-1)}(r^2).\\
\end{eqnarray*}
with $v\in\mR^{m-2|2n-2}$. The boundary term at infinity vanished because of equation \eqref{randinfty}. The boundary term at zero vanishes because of the subsequent lemma \ref{rand0} for $g(r)=f^{(n-1)}(r^2)$ and $k=m$. This partial integration can be iterated, the boundary terms will always be zero because of equation \eqref{randinfty} and lemma \ref{rand0}. In case $M\in-2\mN$ this leads to (for $m=2$ this is obtained immediately)
\begin{eqnarray*}
\int_{\mR^{m|2n}}f(R^2)&=&\int_{\mR^{2|2-M}}f(R^2)\\
&=&\frac{2\pi}{\pi^{1-\frac{M}{2}}}(-1)^{(1-M/2)}\int_{0}^\infty r(\frac{1}{2r}\frac{\partial}{\partial r})^{1-M/2} f(r^2)dr\\
&=&(-\pi)^{\frac{M}{2}}f^{(-M/2)}(0).\\
\end{eqnarray*}
In case $M\in-2\mN-1$ we get
\begin{eqnarray*}
\int_{\mR^{m|2n}}f(R^2)&=&\int_{\mR^{1|1-M}}f(R^2)\\
&=&\frac{2\pi^{1/2}}{\Gamma(\frac{1}{2})\pi^{\frac{1-M}{2}}}(-1)^{(1-M)/2}\int_{0}^\infty  (\frac{1}{2r}\frac{\partial}{\partial r})^{(1-M)/2} f(r^2)dr\\
&=&\frac{2}{\pi^{\frac{1-M}{2}}}(-1)^{(1-M)/2}\int_{0}^\infty  (\frac{1}{2r}\frac{\partial}{\partial r})^{(1-M)/2} f(r^2)dr\\
&=&2(-\pi)^{\frac{M-1}{2}}\int_{0}^\infty  f^{(\frac{1-M}{2})}(r^2)dr.
\end{eqnarray*}
This proves the theorem.
\end{proof}

\begin{lemma}
\label{rand0}
If for a differentiable function $g(r)$ the integral $\int_0^\infty r^{k-2}\frac{\partial}{\partial r}g(r) d r$ with $k\in\mN$ and $k>2$ exists, then the following relation holds:
\begin{eqnarray*}
\lim_{r\to_{>}0}r^{k-2}g(r)&=&0.
\end{eqnarray*}
\end{lemma}
\begin{proof}
First we prove that $\lim_{r\to_{>}0}r^{k-1}\frac{\partial}{\partial r}g(r)=0$. Assume that $\lim_{r\to_{>}0}r^{k-1}\frac{\partial}{\partial r}g(r)=a>0$ (the case $a<0$ is similar). This implies that for an $\epsilon<a$ there is a $\delta_{\epsilon}>0$ such that
\begin{eqnarray*}
r^{k-2}\frac{\partial}{\partial r}g(r)&>&\frac{a-\epsilon}{r},
\end{eqnarray*}
for $r<\delta_{\epsilon}$. This implies that $\int_0^{\delta_{\epsilon}} r^{k-2}\frac{\partial}{\partial r}g(r)$ is infinite. So we obtain $\lim_{r\to_{>}0}r^{k-1}\frac{\partial}{\partial r}g(r)=0$. 

Now, assume that $\lim_{r\to_{>}0}r^{k-2}g(r)=b>0$ with $b$ finite (again the case $b<0$ is treated similarly). This implies that, using the previous result,
\begin{eqnarray*}
\lim_{r\to_{>}0}r\frac{\partial}{\partial r}r^{k-2}g(r)&=&(k-2)b,
\end{eqnarray*}
so for $\eta<(k-2)b$, there is a $\delta_{\eta}>0$ such that for $r<\delta_{\eta}$
\begin{eqnarray*}
\frac{\partial}{\partial r}r^{k-2}g(r)&>&\frac{(k-2)b-\eta}{r}.
\end{eqnarray*}
Integrating this expression yields
\begin{eqnarray*}
r^{k-2}g(r)-r_0^{k-2}g(r_0)&>&\left((k-2)b-\eta\right)\ln(r/r_0)
\end{eqnarray*}
when $r_0<r<\delta_{\eta}$. However, this would imply that 
\begin{eqnarray*}
r^{k-2}g(r)&>&b+\lim_{r_0\to 0}\left((k-2)b-\eta\right)\ln(r/r_0)
\end{eqnarray*}
for $0<r<\delta_{\eta}$, which is impossible.

Finally assume that $\lim_{r\to_{>}0}r^{k-2}g(r)=+\infty$. This implies that for every $K>0$ there is a $\delta_K$ such that $g(r)>K/r^{k-2}$ when $r<\delta_K$. Since $\lim_{r\to_>0}r^{k-1}\frac{\partial}{\partial r}g(r)=0$, for every $\epsilon>0$, there is an $\delta_\epsilon>0$ such that $|g'(r)|<\epsilon/r^{k-1}$ for $r<\delta_{\epsilon}$. For $r<r_0<\delta_{\epsilon}$, this leads to 
\begin{eqnarray*}
-\int_{r}^{r_0}g'(u)du&<&\int_{r}^{r_0}\frac{\epsilon}{u^{k-1}}du,
\end{eqnarray*}
which implies
\begin{eqnarray*}
g(r)<\frac{\epsilon/(k-2)}{r^{k-2}}+g(r_0)-\frac{\epsilon/(k-2)}{r_0^{k-2}}.
\end{eqnarray*}
 When we choose $K>\epsilon/(k-2)$, $g(r)>K/r^{k-2}$ for all $r<\delta_K$ leads to a contradiction. The only possibility that remains is  $\lim_{r\to_{>}0}r^{k-2}g(r)=0$.
\end{proof}

When $M\le 0$, theorem \ref{connectieBer} is not applicable. Integral theorem \ref{dimredint} shows how a regularized version of theorem \ref{connectieBer} can be obtained in the symmetric case $m=2n$. In that case, theorem \ref{dimredint} states
\begin{eqnarray*}
\int_{\mR^{m|2n}}f(R^2)&=&f(0).
\end{eqnarray*}
This is equal to
\begin{eqnarray*}
\int_{\mR^{2n|2n}}f(R^2)&=&-\int_0^\infty \frac{\partial}{\partial v}f(v^2)dv\\
&=&-\lim_{\varepsilon\to_> 0}\frac{2\pi^{\varepsilon/2}}{2\Gamma(\varepsilon/2+1)}\int_0^\infty v^{\epsilon}\frac{\partial}{\partial v}f(v^2)dv\\
&=&\lim_{\varepsilon\to_> 0}\frac{2\pi^{\varepsilon/2}}{\Gamma(\varepsilon/2)}\int_0^\infty  v^{\epsilon-1}f(v^2)dv.\\
\end{eqnarray*}
So we obtain a regularized version of theorem \ref{connectieBer} for orthosymplectically invariant functions for super-dimension $M=0$.

\section{Fundamental solutions and mean value theorem}

The fundamental solutions for the super Laplace operator and its powers were calculated in \cite{DBS6}. In \cite{CDBS1} and \cite{DBSCauchy} they were used to construct Cauchy integral formulae on superspace. They are the solution of the equation
\begin{eqnarray*}
\nabla^{2l}f(\bold{x})&=&(-1)^l\delta(\bold{x})=(-1)^l\delta(\ux)\frac{\pi^n\uxb^{2n}}{n!}.
\end{eqnarray*}
It can be easily seen from the definition of the Berezin integral (\ref{Berezin}) that $\frac{\pi^n\uxb^{2n}}{n!}$ is the fermionic Dirac distribution. As a finite dimensional vector space a Grassmann algebra is isomorphic to its dual. An isomorphism $\chi:\Lambda_{2n}\to\Lambda_{2n}'$ is given by
\begin{eqnarray*}
\chi\left({x\grave{}}_A\right)[{x\grave{}}_B]&=& \int_{B}{x\grave{}}_A{x\grave{}}_B.
\end{eqnarray*}
In this sense we find that the super Dirac distribution has homogeneous (super)degree $-m+2n=-M$. The following lemma states that this uniquely characterizes the Dirac distribution together with the fact that its support is the origin, similar to the bosonic case.
\begin{lemma} 
\label{Dirac}
The super Dirac distribution is up to a multiplicative constant the only Schwartz distribution with homogenous degree $-M$ and with support in the origin of $\mR^m$. 
\end{lemma}
\begin{proof}
It is well-known that a bosonic distribution with support in the origin is a linear combination of derivatives of the Dirac distribution. The Dirac distribution has homogeneous degree $-m$ and each derivative lowers the degree by one. Any super distribution with support in the origin and degree $-M$ is of the form
\begin{eqnarray*}
\sum_{A}d_A(\ux){x\grave{}}_A
\end{eqnarray*}
with $d_A(\ux)$ a homogeneous bosonic distribution of degree $-M-|A|$ with support in the origin. This only exists when $-M-|A|\le -m$, so $|A|=2n$ and $d_{1,\cdots,1}$ is equal to the bosonic Dirac distribution up to a multiplicative constant.
\end{proof}

The fundamental solutions turn out to be spherically symmetric functions. In stead of proving that the expressions obtained in \cite{DBS6} are spherically symmetric we start from new expressions and prove that they are fundamental solutions. First we repeat the purely bosonic fundamental solutions for the Laplace operator and its natural powers. The functions
\begin{equation}
\label{3fundoplA}
\nu_{2j}^{m}(\ux)=
\begin{cases}
r^{2j-m}/\gamma_{j-1,m}& m\;odd\\
r^{2j-m}/\gamma'_{j-1,m}& m\; even\; and \; j<\frac{m}{2}\\
-r^{2j-m}\log (r)/\gamma'_{j-1,m}&m\; even\; and \; j\ge \frac{m}{2}
\end{cases}
\end{equation}
satisfy $\nabla_b^{2j}\nu_{2j}^m(\ux)=(-1)^j\delta(\ux)$. The constants $\gamma$ are explicitly given in \cite{MR745128}. For $m$ odd for instance, the constant is given by
\begin{eqnarray}
\label{gammaodd}
\gamma_{l,m}&=&2^{2l+1}l!\frac{2\pi^{m/2}}{\Gamma(m/2-l-1)}.
\end{eqnarray}
The solutions only depend on $r$ and from now on we write $\nu_{2j}^m(r^2)=\nu_{2j}^m(\ux)$. 

\begin{theorem}
\label{fundopl}
When $M\not\in-2\mN$ and using the notation of definition \ref{defftien}, a fundamental solution for $\nabla^{2l}$ is given by 
\begin{eqnarray*}
\nu_{2l}^{m|2n}(\bold{x})&=&\pi^n2^{2n}\frac{(n+l-1)!}{(l-1)!}\nu_{2l+2n}^{m}({R^2})
\end{eqnarray*}
with $\nu_{2l+2n}^{m}(r^2)$ as in (\ref{3fundoplA}). This means a fundamental solution is given by 
\begin{equation}
\label{3fundopl}
\nu_{2l}^{m|2n}(\bold{x})=
\begin{cases}
R^{2l-M}/\gamma_{l-1,M}& M\;odd\\
R^{2l-M}/{\gamma'_{l-1,M}}& M\; even\; and \; l<\frac{M}{2}\\
-R^{2l-M}\log(R)/({\gamma'_{l-1,M}})&M\; even\; and \; l\ge \frac{M}{2}.
\end{cases}
\end{equation}
\end{theorem}

\begin{proof}
First of all, the proposed formulae are well-defined, as all the $\nu_{2j}^m(r^2)$ are elements of $C^{\infty}(\mR^+)$. First we assume $M$ even. Because then $M>0$ we find that $\nu_{2l+2n}^m(r^2)$ equals $\nu^{M}_{2l}(r^2)$ up to a multiplicative constant. Together with equation \eqref{Laplbosrad} for $M$ bosonic variables this yields 
\begin{eqnarray*}
\left(\frac{\partial^2}{\partial r^2}+\frac{M-1}{r}\frac{\partial}{\partial r}\right)^l\nu^m_{2n+2l}(r^2)&=&0\qquad \mbox{for }r\not=0.\\
\end{eqnarray*}
Now in case $M$ is odd we cannot make the substitution to $M$ bosonic variables, because $M$ could be negative. However, a short calculation shows that for $r\not=0$,
\begin{eqnarray*}
\left(\frac{\partial^2}{\partial r^2}+\frac{M-1}{r}\frac{\partial}{\partial r}\right)r^{2l-M}&=&(2l-M)(2l-2)r^{2(l-1)-M}.
\end{eqnarray*}
Combining the two results we find that for every $M\not\in-2\mN$ and for $r\not=0$
\begin{eqnarray*}
\left(\frac{\partial^2}{\partial r^2}+\frac{M-1}{r}\frac{\partial}{\partial r}\right)^l\nu^m_{2n+2l}(r^2)&=&0.
\end{eqnarray*}
Lemma \ref{deltaftie} and formula \eqref{algmor} then yield
\begin{eqnarray*}
\nabla^{2}\nu_{2l}^{m|2n}(\bold{x})&\sim&4(\nu^m_{2l+2n})^{(2)}(R^2)R^2+2M(\nu^m_{2l+2n})^{(1)}(R^2)\\
&=&h(R^2)\qquad \mbox{with }h(r^2)=\left(\frac{\partial^2}{\partial r^2}+\frac{M-1}{r}\frac{\partial}{\partial r}\right)\nu^m_{2n+2l}(r^2),
\end{eqnarray*}
so
\begin{eqnarray*}
\nabla^{2j}\nu_{2l}^{m|2n}(\bold{x})&\sim&h(R^2)\qquad \mbox{with }h(r^2)=\left(\frac{\partial^2}{\partial r^2}+\frac{M-1}{r}\frac{\partial}{\partial r}\right)^j\nu^m_{2n+2l}(r^2).
\end{eqnarray*}
This implies $\nabla^{2l}\nu_{2l}^{m|2n}(\bold{x})=0$ for $|\ux|>0$. So $\nabla^{2l}\nu_{2l}^{m|2n}(\bold{x})$ is a distribution with support in the origin and $\mE\nu_{2l}^{m|2n}(\bold{x})=(2l-M)\nu_{2l}^{m|2n}(\bold{x})$. From this and lemma \ref{Dirac} we can conclude that $\nabla^{2l}\nu_{2l}^{m|2n}=C_l\delta (\bold{x})$ as distributions for some constant $C_l$. It is now easy to check that the normalizations given in the theorem are correct, either by comparing with the results in \cite{DBS6} or by an explicit calculation. We demonstrate this for $M$ odd. The only part in $\nabla^{2l}R^{2l-M}$ which is not zero is the part by $\uxb^{2n}$, so
\begin{eqnarray*}
\nabla^{2l}R^{2l-M}&=&\nabla_b^{2l}(-1)^{n}\frac{\uxb^{2n}}{n!}(l-\frac{m}{2}+1)\cdots(l-\frac{M}{2})r^{2l-m}\\
&=&\frac{\uxb^{2n}}{n!}\gamma_{l-1,m}\delta(\ux)\frac{\Gamma(\frac{m}{2}-l)}{\Gamma(\frac{M}{2}-l)}\\
&=&\delta(\bold{x})\frac{\gamma_{l-1,m}}{\pi^n}\frac{\Gamma(\frac{m}{2}-l)}{\Gamma(\frac{M}{2}-l)}=\gamma_{l-1,M}\delta(\bold{x}),
\end{eqnarray*}
which proves the normalization constant is correct.
\end{proof}

\begin{remark}
From the definition $\delta(\bold{x})=\delta(\ux)\frac{\pi^n\uxb^{2n}}{n!}$ it is not immediately clear that the super Dirac distribution is an $\mathfrak{osp}(m|2n)$-invariant distribution. However, one can calculate that $L_{ij}\delta(\bold{x})=0$ in distributional sense. The orthosymplectic invariance also follows immediately from the fact that the Dirac distribution is equal to the Laplace operator acting on an orthosymplectically invariant function.
\end{remark}

Before we can establish a mean value theorem for super harmonic functions, we need the following lemma.
\begin{lemma}
\label{lemmaintstelling}
Let $f$ and $g$ be superfunctions which are respectively $n$ and $n+1$  times continuously differentiable in an open neighborhood of the unit sphere $\mS^{m-1}$ and for wich $\nabla f$, $\nabla g$ and $\nabla^2 g$ exist almost everywhere and are integrable over $\mB^m$. Moreover, let $f$ be even in the Grassmann variables. Then the formulas
\begin{align*}
(i)&\int_{SB}\nabla^2g=\int_{SS}\mE g \\
(ii)&\int_{SB}\langle \nabla f,\nabla g\rangle=\int_{SS}f(\mE g)-\int_{SB}f\nabla^2 g
\end{align*}
hold.
\end{lemma}
\begin{proof}
Starting from lemma \ref{Cauchysphere} we find
\begin{eqnarray*}
\int_{SB}\langle \nabla,\nabla g\rangle&=&\int_{SS}\langle \bold{x},\nabla g\rangle
\end{eqnarray*}
which proves the first part because of formula (\ref{Euler}). Since $f$ is even, $\langle \nabla ,f\nabla g\rangle=\langle \nabla f,\nabla g\rangle+f\langle \nabla,\nabla g\rangle$, so the second part follows again from lemma \ref{Cauchysphere}
\begin{eqnarray*}
\int_{SB}\langle \nabla ,f\nabla g\rangle&=&\int_{SS}f\langle \bold{x},\nabla g\rangle.
\end{eqnarray*} 
\end{proof}

In \cite{DBS5} the mean value theorem was proven for harmonic polynomials. From the Pizzetti formula (\ref{Pizzetti}) we immediately obtain that for a harmonic polynomial $h$,
\begin{eqnarray*}
\int_{SS}h&=&\frac{2\pi^{M/2}}{\Gamma(M/2)}h(0).
\end{eqnarray*}
Using formula (\ref{integSS}), we can now prove this property for arbitrary harmonic superfunctions.

\begin{theorem}{(Mean value)}
\label{mean}
Let $\Omega$ be an open set in $\mR^m$ such that $\mB^m\subset\Omega$. Let $h(\bold{x})$ be an element of $C^2(\Omega)_{m|2n}$, satisfying $\nabla^2 h=0$ in $\Omega$. Then one has
\begin{eqnarray*}
\int_{SS}h(\bold{x})&=&\frac{2\pi^{M/2}}{\Gamma(M/2)}h(0).
\end{eqnarray*}
\end{theorem}

\begin{proof}
In \cite{DBS6} it was proven that the bosonic parts, in expansion \eqref{superftie} of a super harmonic function, are polyharmonic and hence elements of $C^\infty(\Omega)$. This means $\int_{SS}h(\bold{x})$ is well-defined. First we consider the case $M\not\in-2\mN$. We start from theorem 11 in \cite{CDBS1}, which states that for a superfunction $g\in C^1(\Omega)_{m|2n}$ with $\mB^m\subset \Omega$,
\begin{eqnarray*}
\int_{SS}\langle \nabla\nu_2^{m|2n}(\bold{x}),\bold{x} g(\bold{x})\rangle-\int_{SB}\langle \nabla \nu_2^{m|2n}(\bold{x}),\nabla g(\bold{x})\rangle&=&-g(0).
\end{eqnarray*}
By rewriting this, substituting $h$ for $g$, using lemma \ref{lemmaintstelling} $(ii)$ and 
\[\langle \nabla\nu_2^{m|2n}(\bold{x}),\bold{x} \rangle=\sum_{k=1}^{m+2n}(-1)^{[k]}X_k\nabla^k\nu_2^{m|2n}(\bold{x})=\sum_{k=1}^{m+2n}X^k\nabla_k\nu_2^{m|2n}(\bold{x})\]
we obtain
\begin{eqnarray*}
\int_{SS}\left(\mE \nu_2^{m|2n}(\bold{x})\right) h(\bold{x})+\int_{SB}\nu_2^{m|2n}(\bold{x})\nabla^2 h(\bold{x})-\int_{SS}\nu_2^{m|2n}(\bold{x})\mE h(\bold{x})&=&-h(0).
\end{eqnarray*}
Now $\nabla^2h=0$, and because $\nu_2^{m|2n}(\bold{x})=\mu(R^2)$ for some $\mu$ (theorem \ref{fundopl}) and lemma \ref{xkwad}, we find
\begin{eqnarray*}
\int_{SS}\left(\mE \nu_2^{m|2n}(\bold{x})\right) h(\bold{x})-\mu(1)\int_{SS}\mE h&=&-h(0).
\end{eqnarray*}
Lemma \ref{lemmaintstelling} $(i)$ implies that $\int_{SS}\mE h=0$, so finally
\begin{eqnarray*}
\int_{SS}\left(\mE \nu_2^{m|2n}(\bold{x})\right) h(\bold{x})&=&-h(0).
\end{eqnarray*}
As $\mE\nu_2^{m|2n}(\bold{x})$ is still a spherically symmetric function (see (\ref{Eulerspherical})) we conclude that
\begin{eqnarray*}
\int_{SS}h(\bold{x})&=&Ch(0),
\end{eqnarray*}
for some fixed constant. This constant has to be $\frac{2\pi^{M/2}}{\Gamma(M/2)}$, because of the result for polynomials. This constant can also be calculated from theorem \ref{fundopl}. We show this for the case $M$ odd, yielding
\begin{eqnarray*}
\mE\nu_{2}^{m|2n}&=&2(1-M/2)\frac{1}{2\frac{2\pi^{M/2}}{\Gamma(M/2-1)}}R^{2-M}\\
&=&-\frac{\Gamma(M/2)}{2\pi^{M/2}}R^{2-M}.
\end{eqnarray*}

In case $M\in-2\mN$ we observe that the right-hand side is zero since the surface of the supersphere $\sigma_M$ is zero. So we have to prove that $\int_{SS}h(\bold{x})=0$ for $h$ harmonic. When $M\in-2\mN$, $R^{2-M}$ is harmonic (lemma \ref{deltaftie}), also in the origin since $2-M>0$. As $\nabla^2$ is an even operator, each harmonic function $h$ is of the form $h_++h_{-}$ with $h_+$ even, $h_-$ odd in the Grassmann variables and both $h_+$ and $h_-$ harmonic. Using the fact that $R^{2-M}$ is both harmonic and spherically symmetric we can use lemma \ref{lemmaintstelling} $(ii)$ twice and lemma \ref{lemmaintstelling} $(i)$ to calculate
\begin{eqnarray*}
\int_{SS}h_+(\bold{x})&=&\int_{SS}h_+R^{2-M}=\frac{1}{2-M}\int_{SB}\langle\nabla h_+,\nabla R^{2-M}\rangle=\frac{1}{2-M}\int_{SB}\langle\nabla R^{2-M},\nabla h_+\rangle\\
&=&\frac{1}{2-M}\int_{SS} R^{2-M}(\mE h_+)=\frac{1}{2-M}\int_{SS}(\mE h_+)=0.
\end{eqnarray*}
For the odd case we calculate similarly, using lemma \ref{Cauchysphere},
\begin{eqnarray*}
\int_{SS}h_-(\bold{x})&=&\int_{SS}h_-R^{2-M}=\frac{1}{2-M}\int_{SS}h_-\mE R^{2-M}\\
&=&\frac{1}{2-M}\int_{SS}\sum_kX^k(-1)^{[k]}h_-\nabla_kR^{2-M}=\frac{1}{2-M}\int_{SB}\sum_k\nabla^k(-1)^{[k]}h_-\nabla_kR^{2-M}\\
&=&\frac{1}{2-M}\int_{SB}\sum_k(-1)^{[k]}(\nabla^kh_-)\nabla_kR^{2-M}+\frac{1}{2-M}\int_{SB}\sum_kh_-\nabla^k\nabla_kR^{2-M}\\
&=&\frac{1}{2-M}\int_{SB}\sum_k(-1)^{[k]}(-1)^{[k]([k]+1)}(\nabla_kR^{2-M})(\nabla^kh_-)\\
&=&\frac{1}{2-M}\int_{SB}\langle\nabla R^{2-M},\nabla h_- \rangle=\frac{1}{2-M}\int_{SS} R^{2-M}(\mE h_-)\\
&=&\frac{1}{2-M}\int_{SS}(\mE h_-)=0
\end{eqnarray*}
which concludes the proof.
\end{proof}

When $M\in-2\mN$ the integral of a harmonic polynomial on the supersphere vanishes because there is a function $f(R^2)$ which is harmonic on $\mB^m$. This is clearly only possible for negative even dimensions. Moreover, it is the same reason why there is no Fischer decomposition in these cases (see lemma \ref{scalFischer}).

\section{Orthosymplectically invariant quantum problems}
In this section we show how spherically symmetric Schr\"odinger equations in superspace can be solved using harmonic analysis, orthosymplectically invariant functions and dimensional reduction. We can solve the case $M>0$ immediately using the purely bosonic case. Also when $M$ is negative these techniques can be used to create an Ansatz for solutions leading to a one-dimensional differential equation.

When spherically symmetric quantum hamiltonians (such as the harmonic oscillator \cite{DBS3} or the hydrogen atom \cite{MR2395482}) are generalized to superspace, super hamiltonians with an $\mathfrak{osp}(m|2n)$ invariance are obtained. In \cite{MR2395482} the energy eigenvalues and corresponding eigenspaces were determined for the quantum Kepler problem (hydrogen atom) in superspace. This was realized by constructing an $\mathfrak{osp}(2,m+1|2n)$ dynamical supersymmetry for the system. 

In \cite{DBS3} a basis of Clifford-Hermite functions was constructed for the quantum harmonic oscillator in superspace. In the bosonic case, the basis of Clifford-Hermite functions follows from introducing spherical coordinates and using the fact that $\nabla_b^2$, $r^2$ and $\mE_b+m/2$ generate the $\mathfrak{sl}_2$ Lie algebra. In \cite{MR2395482} this $\mathfrak{sl}_2$ algebra (now generated by $\nabla^2$ and  $R^2$) was transformed into the $\mathfrak{so}(2,1)$ algebra, which led together with the $\mathfrak{osp}(m|2n)$-symmetry to $\mathfrak{osp}(2,m+1|2n)$.

In this section we will show how we can use spherically symmetric functions and the $\mathfrak{sl}_2$-properties to solve more general super hamiltonians with orthosymplectic invariance. A general hamiltonian is of the form 
\begin{eqnarray*}
H&=&-\frac{1}{2}\nabla^2+V(\bold{x})
\end{eqnarray*}
with $V$ a general superfunction. Theorem \ref{uniciteit} implies that all hamiltonians with an $\mathfrak{osp}(m|2n)$ invariance have a spherically symmetric function (definition \ref{defftien}) as potential,
\begin{eqnarray}
\label{superham}
H&=&-\frac{1}{2}\nabla^2+V(R^2).
\end{eqnarray}

When we expand the superfunction $f$ in the equation $Hf=Ef$ as in formula \eqref{superftie} we obtain a system of partial differential equations for the functions $f_A$ on $\mR^m$. When $M>0$, the bosonic hamiltonian of the form 
\begin{eqnarray}
\label{bosham}
H_b&=&-\frac{1}{2}\nabla_b^2+V(r_{\uy}^2)
\end{eqnarray}
in $M$ bosonic dimensions, $\uy\in\mR^M$ and $\nabla_b^2=\sum_{j=1}^M\partial_{y_j}^2$ is a special case.

To study the completeness of the solutions of the hamiltonian a suitable Hilbert space structure should be considered. In a forthcoming article the authors will construct such a Hilbert space and show that every suitable Hilbert space contains generalized functions. These generalized functions already appear in the following theorem.

\begin{theorem}
\label{superopl}
If for the case with $M>0$ bosonic variables the function $H_k^bf(r^2_{\uy})\in L_2(\mR^M)$, with $H_k^b\in\cH_{k,M}^b$, is an eigenvector of the $M$-dimensional hamiltonian (\ref{bosham}) with eigenvalue $E$, then
\begin{eqnarray*}
\left[-\frac{1}{2}\nabla^2+V(R^2)\right]H_k f(R^2)&=&EH_kf(R^2),
\end{eqnarray*}
for every $H_k\in\cH_k$ in superspace $\mR^{m|2n}$ with $M=m-2n$. In case $f(r^2_{\uy})$ is $n$ times differentiable $f(R^2)$ is defined in the usual way, if not, $f(R^2)$ is defined formally as as an element of $\cS'(\mR^{m})\otimes\Lambda_{2n}$.
\end{theorem}

\begin{proof}
We start from
\begin{eqnarray*}
\left[-\frac{1}{2}\nabla_b^2+V(r^2_{\uy})\right]H_k^bf(r^2_{\uy})&=&EH_k^bf(r^2_{\uy}).
\end{eqnarray*}
Substituting equation (\ref{laplradharm}) for $M$ bosonic variables yields
\begin{eqnarray*}
\left[-2r^2_{\uy}f^{(2)}(r^2_{\uy})-(2k+M)f^{(1)}(r^2_{\uy})+V(r^2_{\uy})f(r^2_{\uy})\right]H_k^b&=&Ef(r^2_{\uy})H_k^b.
\end{eqnarray*}
From this it is clear that the exact form of the spherical harmonic is irrelevant. By using property (\ref{algmor}) we obtain
\begin{eqnarray*}
-2R^2f^{(2)}(R^2)-(2k+M)f^{(1)}(R^2)+V(R^2)f(R^2)&=&Ef(R^2)
\end{eqnarray*}
for any supervector $\bold{x}$ and $R^2=\langle \bold{x},\bold{x}\rangle$. In case $f$ is not sufficiently smooth, this equation should be regarded in distributional sense. If $\bold{x}$ is a supervector with super-dimension $M$, equation (\ref{laplradharm}) yields
\begin{eqnarray*}
\left[-\frac{1}{2}\nabla^2+V(R^2)\right]H_k f(R^2)&=&EH_kf(R^2).
\end{eqnarray*}
\end{proof}

These solutions will turn out to be a complete set once the correct Hilbert space is considered.

To find the multiplicities for the energy levels, it is important to note that the dimension of the space of spherical harmonics of degree $k$ is given by (see \cite{DBS5})
\begin{eqnarray*}
\dim\cH_k&=&\sum_{i=0}^{\min(k,2n)}\binom{2n}{i}\binom{k-i+m-1}{m-1}-\sum_{i=0}^{\min(k-2,2n)}\binom{2n}{i}\binom{k-i+m-3}{m-1}.
\end{eqnarray*}

\section{Funk-Hecke theorem in superspace}

In theorem 7 in \cite{DBS5}, a Funk-Hecke theorem on superspace for polynomials was proven. The classical (bosonic) version of the Funk-Hecke theorem can e.g. be found in \cite{MR0499342, MR1151617}. For $\ux,\uy\in\mR^m$, a continuous function $\phi(\langle \ux,\uy\rangle)$, $H_l^b\in\cH_{l,m}^b$ and $m>1$,
\begin{eqnarray*}
\int_{\mS^{m-1}}\phi(\langle \ux,\uy\rangle)H_{l}^b(\ux)&=&\alpha_{m,l}[\phi](r_y^2)\frac{H_l(\uy)}{r_y^l}
\end{eqnarray*}
holds, with
\begin{eqnarray}
\label{alphar}
\alpha_{m,l}[\phi](r_y^2)&=&\sigma_{m-1}\int_{-1}^1\phi(r_yt)P_l^m(t)(1-t^2)^{\frac{m-3}{2}} dt.
\end{eqnarray}
$P^M_l$ are the Jacobi polynomials $P_l^{\frac{M-3}{2},\frac{M-3}{2}}$, also known as the Gegenbauer polynomials, 
\begin{eqnarray*}P^M_l(t)&=&\frac{1}{\binom{l+M-3}{l}}C^{\frac{M-2}{2}}_l(t).
\end{eqnarray*}

Using the supersphere integration from \cite{CDBS1} it will be possible to extend the existing Funk-Hecke theorem in superspace to more general functions. This is the subject of the current section. We first repeat the result from \cite{DBS5}.

\begin{lemma}
Let $\bold{x},\bold{y}$ be independent super vector variables. Let $H_l \in \cH_l$ and $k\in\mN $, then
\begin{eqnarray}
\label{FHpol}
\int_{SS} \langle \bold{x},\bold{y} \rangle^k H_l(\bold{x}) &=& \alpha_{M,l}[t^k]R_y^{k-l}H_l(\bold{y}) 
\end{eqnarray}
with 
\begin{eqnarray*}
\label{cteFH}
\alpha_{M,l}[t^k] &=& \frac{k!}{(k-l)!} \frac{2 \pi^{\frac{M-1}{2}}}{2^l} \frac{\Gamma(\frac{k-l+1}{2})}{\Gamma(\frac{M+k+l}{2})}
  \quad \mbox{if $k+l$ even and $k \geq l$}\\
&=&0 \quad \mbox{if $k+l$ odd}\\
&=&0 \quad \mbox{if $k<l$}
\end{eqnarray*}
or
\begin{eqnarray*}
\alpha_{M,l}[t^k] &=&\sigma_{M-1}\int_{-1}^1t^kP_l^M(t)(1-t^2)^{\frac{M-3}{2}} dt
  \quad \mbox{if  $M > 1$}.
\end{eqnarray*}
\label{FunkHeckeSuperspace1}
\end{lemma}

From the definition of $\alpha_{M,l}[t^k]$ it follows that we obtain only even positive powers of $R_y$ in equation \eqref{FHpol}. For general polynomials $p\in\mR[t]$, $\alpha_{M,l}[p(t)]$ is defined by linearity, so the theorem holds for polynomials in $\langle \bold{x},\bold{y}\rangle$.

General functions depending on $\langle \bold{x},\bold{y}\rangle$ can be defined using definition \ref{defftien2}. This means that for an interval $[-a,a]$ and a function $\phi\in C^{2n}([-a,a])$, we put
\begin{eqnarray*}
\phi(\langle \bold{x},\bold{y}\rangle)&=&\sum_{j=0}^{2n}\frac{\langle \uxb,\uyb\rangle^j}{j!}\phi^{(j)}(\langle\ux,\uy\rangle),
\end{eqnarray*}
as an element of $C(\Sigma_a)\otimes\Lambda_{4n}^{\bold{x},\bold{y}}$, with
\begin{eqnarray}
\Sigma_a&=&\{(\ux,\uy)\in\mR^m\times\mR^m|\langle \ux,\uy\rangle\in [-a,a]\}.
\end{eqnarray}
We call such functions super zonal functions. Similarly to theorem \ref{uniciteit} of the spherically symmetric functions we can prove that these are the only superfunctions of two vector variables $(\bold{x},\bold{y})$ which are solutions of the equations
\begin{eqnarray*}
\left(X_i\partial_{X^j}-(-1)^{[i][j]}Y_j\partial_{Y^i}\right)f&=&0,
\end{eqnarray*}
for all $i,j$.

Now we start to prove the Funk-Hecke theorem in superspace for general zonal functions. 
\begin{lemma}
\label{lemmaFH}
Let $\phi\in C^{2n}([-a,a])$. Let $\{p_j\}$, $j\in \mN$ be a sequence of polynomials which converges uniformly to $\phi$, together with their first $2n$ derivatives, i.e.
\begin{eqnarray*}
p_j^{(i)}(t)&\to &\phi^{(i)}(t), \qquad 0\le i\le 2n
\end{eqnarray*}
on $[-a,a]$. Then for a super spherical harmonic $H_l\in\cH_l$
\begin{eqnarray*}
\lim_{j\to\infty}\int_{SS,\bold{x}}p_j^{}(\langle \bold{x},\bold{y}\rangle ) H_l(\bold{x})&=& \int_{SS,\bold{x}}\phi^{}(\langle \bold{x},\bold{y}\rangle)H_l(\bold{x})
\end{eqnarray*}
pointwise as functions of $\uy\in\overline{\mB^m(a)}$ (with $\overline{\mB(a)}$ the closure of the ball with radius $a$ in $\mR^m$) taking values in the Grassmann algebra $\Lambda_{2n}^y$.
\end{lemma}

\begin{proof}
For a general zonal function $\Psi(\langle \bold{x},\bold{y}\rangle)$ (so also for a polynomial), we can write the integration over the supersphere (\ref{integSS}) as
\begin{eqnarray*}
\int_{SS,\bold{x}}\Psi(\langle \bold{x},\bold{y}\rangle)H_l&=&\sum_{2j+k\le 2n}\int_{\mS^{m-1}}d\underline{\xi}\int_B\frac{\uxb^{2j}\langle \uxb,\uyb\rangle^k}{j!k!}\left[\left(\frac{\partial}{\partial r^2}\right)^jr^{m-2}\Psi^{(k)}(\langle\ux,\uy\rangle)H_l\right]_{r=1}.
\end{eqnarray*}
This means that for $|\uy|\le a$ the convergence follows from the uniform convergence of the first $2n$ derivatives.
\end{proof}

Now we extend the notation from \eqref{alphar}, which is only defined for $m>1$, to all dimensions for polynomials, $\alpha_{M,l}[t^k](u^2)=\alpha_{M,l}(t^k)u^k$. Definition \ref{defftien} then implies that $\alpha_{M,l}[t^k](R_y^2)=\alpha_{M,l}(t^k)R_y^k$.

\begin{corollary}
\label{corFH}
For a sequence $p_j$ and $\phi$ as in lemma \ref{lemmaFH}, one has
\begin{eqnarray*}
\int_{SS,\bold{x}}\phi^{}(\langle \bold{x},\bold{y}\rangle)H_l(\bold{x})&=&\left(\lim_{j\to\infty}\alpha_{M,l}\left[p_j( t)\right](R_y^2)\right)\frac{H_l(\bold{y})}{R_y^l},
\end{eqnarray*}
pointwise for $\uy\in\overline{\mB^m(a)}$.
\end{corollary}
\begin{proof}
This follows immediately from lemma \ref{lemmaFH} together with the Funk-Hecke theorem for polynomials in lemma \ref{FunkHeckeSuperspace1}.
\end{proof}

Corollary \ref{corFH} implies that the limit on the right-hand side in fact converges and does not depend on the choice of $p_j$, so we can define
\begin{definition}
\label{alphaM}
Let $\phi\in C^{2n}([-a,a])$ and $\uy\in \mB^m(a)$. Then $\alpha_{M,l}[\phi](R_y^2)$ is given by
\begin{eqnarray*}
\alpha_{M,l}[\phi](R_y^2)&=&\lim_{j\to\infty }\alpha_{M,l}[p_j](R_y^2),
\end{eqnarray*}
with $p_j(t)$ a sequence of polynomials for which $p_j^{(i)}\to\phi^{(i)}$, for $0\le i\le 2n$, uniformly on $[-a,a]$.
\end{definition}

The following lemma implies that this is well-defined, since for $M>1$, $\alpha_{M,l}[\phi](R_y^2)$ is already defined by equation \eqref{alphar} and definition \ref{defftien}.
\begin{lemma}
\label{interpretatieFH}
For $M>1$, $\phi\in C^{2n}([-a,a])$, $\uy\in\mB^m( a)$ and with $\alpha_{M,l}[\phi](R_y^2)$ as defined in lemma \ref{alphaM}, the following relations hold,
\begin{eqnarray*}
\alpha_{M,l}[\phi](R^2_y)&=&\sigma_{M-1}\int_{-1}^1\phi(R_y\,t)P_l^M(t)(1-t^2)^{\frac{M-3}{2}}\\
&=&\sum_{k=0}^n\frac{(-1)^k\uyb^{2k}}{k!}\left(\alpha_{M,l}[\phi]\right)^{(k)}(r^2).
\end{eqnarray*}
\end{lemma}
\begin{proof}
The first equality follows from
\begin{eqnarray*}
& &\lim_{j\to\infty}\sigma_{M-1}\int_{-1}^1p_j(R_y\,t)P_l^M(t)(1-t^2)^{\frac{M-3}{2}}\\
&=&\lim_{j\to\infty}\sum_{k=0}^n(-1)^k\frac{\uyb^{2k}}{k!}\sigma_{M-1}\int_{-1}^1\left(\frac{\partial}{\partial r_y^2}\right)^kp_j(r_y\,t)P_l^M(t)(1-t^2)^{\frac{M-3}{2}}\\
&=&\sigma_{M-1}\int_{-1}^1\phi(R_y\,t)P_l^M(t)(1-t^2)^{\frac{M-3}{2}}
\end{eqnarray*}
which follows from the uniform convergence of the first $n$ derivatives. The second equality follows from the smoothness of the integrand
\begin{eqnarray*}
\sigma_{M-1}\int_{-1}^1\phi(R_y\,t)P_l^M(t)(1-t^2)^{\frac{M-3}{2}}&=&\sum_{k=0}^n(-1)^k\frac{\uyb^{2k}}{k!}\sigma_{M-1}\int_{-1}^1\left(\frac{\partial}{\partial r_y^2}\right)^k\phi(r_y\,t)P_l^M(t)(1-t^2)^{\frac{M-3}{2}}\\
&=&\sum_{k=0}^n(-1)^k\frac{\uyb^{2k}}{k!}\left(\frac{\partial}{\partial r_y^2}\right)^k\left(\alpha_{M,l}[\phi](r_y^2)\right).
\end{eqnarray*}
\end{proof}

Now we can state the general Funk-Hecke theorem in superspace.

\begin{theorem}
Let $\bold{x},\bold{y}$ be independent vector variables. Let $H_l \in \cH_l$ be a spherical harmonic of degree $l$ and $\phi\in C^{2n}([-a,a])$. Then one has
\[
\int_{SS,\bold{x}} \phi(\langle \bold{x},\bold{y} \rangle) H_l(\bold{x}) = \alpha_{M,l}[\phi](R_y^2) \, \frac{H_l(\bold{y})}{R_y^l}
\]
for $\uy\in \mB^m(a)$, with $\alpha_{M,l}(\phi(R_y\,t))$ as defined in definition \ref{alphaM} or lemma \ref{interpretatieFH} for $M>1$.
\label{FunkHeckeSuperspace2}
\end{theorem}

\begin{proof}
Because $\phi^{(2n)}\in C([-a,a])$ and by the Weierstrass approximation theorem there is a sequence of polynomials $\{q_j\}$ for which $q_j(t)\to\phi^{(2n)}(t)$ uniformly on $[-a,a]$. By integrating $2n$ times we obtain a sequence of polynomials $\{p_j\}$ for which $p_j^{(i)}\to\phi^{(i)}$ uniformly on $[-a,a]$ for $0\le i\le 2n$. The theorem then follows from corollary \ref{corFH}.
\end{proof}

\section{Bochner's relations and the Mehler formula for the super Fourier transform}

\subsection{Bochner's periodicity relations}

Bochner's relations in bosonic analysis give an expression for the Fourier transform of a function $H_k^b\psi(r^2)$, with $H_k^b\in\cH_{k,m}^b$, in terms of the Hankel transform, see e.g. \cite{MR1151617}.
\begin{definition}
The Hankel transform of a function $f\in\cS(\mR^+)$ is given by
\begin{eqnarray*}
\cH_{\nu}[f](u)&=&\int_{0}^\infty f(r)\frac{J_{\nu}(ru)}{(ru)^\nu}r^{2\nu+1}dr,
\end{eqnarray*}
with $J_{\nu}$ the Bessel function of the first kind of order $\nu$ for $\nu>-\frac{1}{2}$.
\end{definition}

For convenience we use the following transformation,
\begin{eqnarray*}
\cF_{\nu}[\psi](u^2)&=&\cH_{\nu}[\psi\circ\Upsilon](u),
\end{eqnarray*}
with $\Upsilon(u)=u^2$. The classical Bochner's relations are given by

\begin{theorem}
For $H_k^b\in\cH_{k,m}^b$, $\psi\in\cS(\mR_+)$ and with $r^2=\langle\ux,\ux\rangle$ and $u^2=\langle \uy,\uy\rangle$, the Fourier transform of $H_k^b\psi(r^2)$ is given by
\begin{eqnarray*}
\cF^\pm_m[H_k^b(\ux)\psi(r^2)](\uy)&=&(\pm i)^kH_k^b(\uy)\cF_{k+\frac{m}{2}-1}[\psi](u^2).
\end{eqnarray*}
\end{theorem}

We will show how the super Fourier transform of an orthosymplectically invariant function, multiplied with a spherical harmonic, can equally be expressed in terms of the classical Hankel transform for $M>0$. We start by proving the following application of the Funk-Hecke theorem in superspace.
\begin{lemma}
\label{SFH}
For $M>1$ and $H_k\in\cH_k$, the following relation holds, with $v\in\mR^+$,
\begin{eqnarray*}
\int_{SS,\bold{x}}\exp(iv\langle \bold{x},\bold{y}\rangle)H_k(\bold{x})&=&i^k(2\pi)^{\frac{M}{2}}(vR_y)^{1-\frac{M}{2}}J_{\frac{M}{2}+k-1}(vR_y)\frac{H_k(\bold{y})}{R_y^k}.
\end{eqnarray*}
\end{lemma}
\begin{proof}
First of all we note that the right-hand side is well-defined, since $J_{\frac{M}{2}+k-1}$ is analytic on $\mR^+$. Using theorem \ref{FunkHeckeSuperspace2} and lemma \ref{interpretatieFH} we obtain
\begin{eqnarray*}
\int_{SS,\bold{x}}\exp(iv\langle \bold{x},\bold{y}\rangle)H_k(\bold{x})&=&\alpha_{M,k}\left[\exp(ivt)\right](R_y^2)\frac{H_k(\bold{y})}{R_y^k}.
\end{eqnarray*}
From the bosonic Bochner's relations we find that for $vu\in\mR^+$, 
\begin{eqnarray*}
\alpha_{M,k}[\exp(ivt)](u^2)=i^k(2\pi)^{\frac{M}{2}}(vu)^{1-\frac{M}{2}}J_{\frac{M}{2}+k-1}(vu)
\end{eqnarray*}
for $M>1$. 
\end{proof}

This leads to the Bochner's periodicity relations in superspace.
\begin{theorem}[Bochner's relations]
\label{Bochner}

For $M>1$, $\psi\in\cS(\mR_+)$, $\psi(R^2)$ as in definition \ref{defftien} and $H_k\in\cH_k$, the following relation holds
\begin{eqnarray*}
\cF^\pm_{m|2n}[H_k(\bold{x})\psi(R^2)](\bold{y})&=&(\pm i)^kH_k(\bold{y})\cF_{k+\frac{M}{2}-1}[\psi](R_y^2),
\end{eqnarray*}
with $\cF_{k+\frac{M}{2}-1}[\psi](R_y^2)$ as in definition \ref{defftien}.
\end{theorem}
\begin{proof}
We calculate the left-hand side using theorem \ref{connectieBer}, lemma \ref{xkwad} and lemma \ref{SFH}, yielding
\begin{eqnarray}
\nonumber
\cF^\pm_{m|2n}[H_k(\bold{x})\psi(R^2)](\bold{y})&=&(2\pi)^{-\frac{M}{2}}\int_{\mR^{m|2n}}\exp(\pm i\langle \bold{x},\bold{y}\rangle)H_k(\bold{x})\psi(R^2)\\
\nonumber
&=&(2\pi)^{-\frac{M}{2}}\int_0^\infty dv\,v^{M-1}\,\int_{SS,\bold{x}}\exp(\pm iv\langle \bold{x},\bold{y}\rangle)v^kH_k(\bold{x})\psi(v^2R^2)\\
\nonumber
&=&(2\pi)^{-\frac{M}{2}}\int_0^\infty dv\,v^{M+k-1}\psi(v^2)(\pm i)^k(2\pi)^{\frac{M}{2}}(vR_y)^{1-\frac{M}{2}}J_{\frac{M}{2}+k-1}(vR_y)\frac{H_k(\bold{y})}{R_y^k}\\
\label{BochnerMehler}
&=&(\pm i)^kH_k(\bold{y})\,\int_0^\infty dv\,v^{M+2k-1}\,(vR_y)^{1-\frac{M}{2}-k}\psi(v^2)\,J_{\frac{M}{2}+k-1}(vR_y).
\end{eqnarray}
Because $t^{-\nu}J_{\nu}(t)$ is analytical in $[0,+\infty[$, we can apply definition \ref{defftien} to obtain
\begin{eqnarray*}
\cF^\pm_{m|2n}[H_k(\bold{x})\psi(R^2)](\bold{y})&=&(\pm i)^kH_k(\bold{y})\,\int_0^\infty dv\,v^{M+2k-1}\,\psi(v^2)\,\sum_{j=0}^n(-1)^j\frac{\uyb^{2j}}{j!}\left(\frac{\partial}{\partial r_y^2}\right)^j\frac{J_{\frac{M}{2}+k-1}(vr_y)}{(vr_y)^{k+\frac{M}{2}-1}}\\
&=&(\pm i)^kH_k(\bold{y})\,\sum_{j=0}^n(-1)^j\frac{\uyb^{2j}}{j!}\left(\frac{\partial}{\partial r_y^2}\right)^j\cF_{\frac{M}{2}+k-1}[\psi](r_y^2).
\end{eqnarray*}
This concludes the proof of the theorem.
\end{proof}

\subsection{Mehler formula}
To obtain a Mehler formula for the super Fourier kernel we start from equation \eqref{BochnerMehler},
\begin{eqnarray*}
& &(2\pi)^{-M/2}\int_{\mR^{m|2n}}\exp(\pm i\langle \bold{x},\bold{y}\rangle)H_k(\bold{x})\psi(R^2)\\
&=&(\pm i)^kH_k(\bold{y})\,\int_0^\infty dv\,v^{M+2k-1}\,(vR_y)^{1-\frac{M}{2}-k}\psi(v^2)\,J_{\frac{M}{2}+k-1}(vR_y).
\end{eqnarray*}
Using equation \eqref{reprkern} and theorem \ref{connectieBer} yields formally
\begin{eqnarray*}
& &(\pm i)^kH_k(\bold{y})\,\int_0^\infty dv\,v^{M+2k-1}\,(vR_y)^{1-\frac{M}{2}-k}\psi(v^2)\,J_{\frac{M}{2}+k-1}(vR_y)\\
&=&\sum_{l=0}^\infty(\pm i)^l\int_{SS,\bold{x}}H_k(\bold{x})F_l(\bold{x},\bold{y})\,\int_0^\infty dv\,v^{M+l+k-1}\,(vR_y)^{1-\frac{M}{2}-l}\psi(v^2)\,J_{\frac{M}{2}+l-1}(vR_y)\\
&=&\sum_{l=0}^\infty(\pm i)^l\,\int_0^\infty dv\,v^{M-1}\,(vR_y)^{1-\frac{M}{2}-l}\psi(v^2)\,J_{\frac{M}{2}+l-1}(vR_y)\int_{SS,\bold{x}}H_k(v\bold{x})F_l(v\bold{x},\bold{y})\\
&=&\sum_{l=0}^\infty(\pm i)^l\,\int_{\mR^{m|2n}} \,(R_xR_y)^{1-\frac{M}{2}-l}\,J_{\frac{M}{2}+l-1}(R_xR_y)F_l(\bold{x},\bold{y})\psi(R_x^2)H_k(\bold{x}).
\end{eqnarray*}
So we obtain an alternative expression for the Fourier kernel, which leads to the Mehler formula
\begin{eqnarray}
\label{MehlerBessel}
(2\pi)^{-M/2}\exp(\pm i\langle \bold{x},\bold{y}\rangle)&=&\sum_{k=0}^\infty (\pm i)^k F_k(\bold{x},\bold{y})\frac{J_{\frac{M}{2}+k-1}(R_xR_y)}{(R_xR_y)^{\frac{M}{2}+k-1}}.
\end{eqnarray}

The convergence of this series can be proven using the exact same technique as in theorem 14 in \cite{CDBS2}. By comparing formula \eqref{MehlerBessel} to the Mehler formula \eqref{superMehler2} we obtain
\begin{eqnarray*}
\frac{J_{\frac{M}{2}+k-1}(R_xR_y)}{(R_xR_y)^{\frac{M}{2}+k-1}}&=&\sum_{j=0}^\infty\frac{2j!(-1)^j}{\Gamma(j+\frac{M}{2}+k)}    L_j^{\frac{M}{2}+k-1}(R^2_x)L_j^{\frac{M}{2}+k-1}(R_y^2) \exp(-\frac{R^2_x + R^2_y}{2}),
\end{eqnarray*}
which is the orthosymplectic analog of the classical Hille-Hardy formula (see e.g. \cite{HH}, p. 189).

\end{document}